\documentclass[a4paper,USenglish]{lipics-v2016}
 
\usepackage{microtype}

\title{Parameterized (Approximate) Defective Coloring\footnote{The authors were supported by the GRAPA – Graph Algorithms for Parameterized Approximation – 38593YJ PHC Sakura Project}}
\titlerunning{Parameterized (Approximate) Defective Coloring} 

\author[1]{Rémy Belmonte\footnote{Supported by the ELC project (Grant-in-Aid for Scientific Research on Innovative Areas, MEXT Japan)}}
\author[2]{Michael Lampis}
\author[3]{Valia Mitsou}
\affil[1]{University of Electro-Communications, Chofu, Tokyo, 182-8585, Japan, \texttt{remy.belmonte@uec.ac.jp}}
\affil[2]{Université Paris-Dauphine, PSL Research University, CNRS, UMR 7243 \\ LAMSADE, 75016, Paris, France, \texttt{michail.lampis@dauphine.fr}}
\affil[3]{Université Paris-Diderot, IRIF, CNRS, 8243, \\ Université Paris-Diderot – Paris 7, 75205, Paris, France, \texttt{vmitsou@irif.fr}}
\authorrunning{R. Belmonte, M. Lampis, and V. Mitsou} 

\Copyright{Rémy Belmonte, Michael Lampis, and Valia Mitsou}


\EventEditors{Rolf Niedermeier and Brigitte Vall\'ee}
\EventNoEds{2}
\EventLongTitle{35th Symposium on Theoretical Aspects of Computer Science (STACS 2018)}
\EventShortTitle{STACS 2018}
\EventAcronym{STACS}
\EventYear{2018}
\EventDate{February 28 to March 3, 2018}
\EventLocation{Caen, France}
\EventLogo{}
\SeriesVolume{96}
\ArticleNo{11} 

\begin{document}

\maketitle

\newcommand{\tw}{\mathrm{tw}}
\newcommand{\pw}{\mathrm{pw}}
\newcommand{\fvs}{\mathrm{fvs}}
\newcommand{\td}{\mathrm{td}}
\newcommand{\vc}{\mathrm{vc}}
\newcommand{\DC}{\textsc{Defective Coloring}} 
\newcommand{\MCC}{$k$-\textsc{Multi-Colored Clique}} 
\newcommand{\D}{\ensuremath\Delta^*} 
\newcommand{\C}{\ensuremath\mathrm{\chi_d}} 

\newcommand{\T}{\ensuremath\mathcal{T}} 

\begin{abstract}

In \DC\ we are given a graph $G=(V,E)$ and two integers $\C,\D$ and are asked
if we can partition $V$ into $\C$ color classes, so that each class induces a
graph of maximum degree $\D$. We investigate the complexity of this
generalization of \textsc{Coloring} with respect to several well-studied graph
parameters, and show that the problem is W-hard parameterized by treewidth,
pathwidth, tree-depth, or feedback vertex set, if $\C=2$. As expected, this
hardness can be extended to larger values of $\C$ for most of these parameters,
with one surprising exception: we show that the problem is FPT parameterized by
feedback vertex set for any $\C\neq 2$, and hence $2$-coloring is the only hard
case for this parameter.  In addition to the above, we give an ETH-based lower
bound for treewidth and pathwidth, showing that no algorithm can solve the
problem in $n^{o(\pw)}$, essentially matching the complexity of an algorithm
obtained with standard techniques. 

We complement these results by considering the problem's approximability and
show that, with respect to $\D$, the problem admits an algorithm which for any
$\epsilon>0$ runs in time $(\tw/\epsilon)^{O(\tw)}$ and returns a solution with
exactly the desired number of colors that approximates the optimal $\D$ within
$(1+\epsilon)$. We also give a $(\tw)^{O(\tw)}$ algorithm which achieves the
desired $\D$ exactly while $2$-approximating the minimum value of $\C$. We show
that this is close to optimal, by establishing that no FPT algorithm can (under
standard assumptions) achieve a better than $3/2$-approximation to $\C$, even
when an extra constant additive error is also allowed.

 \end{abstract}

\newcounter{count:constr} 

\newcounter{count:constr2}

\section{Introduction}

%
%
%

\DC\ is the following problem: we are given a graph $G=(V,E)$, and two integer
parameters $\C,\D$, and are asked whether there exists a partition of $V$ into
at most $\C$ sets (color classes), such that each set induces a graph with
maximum degree at most $\D$. \DC, which is also sometimes referred to in the
literature as \textsc{Improper Coloring}, is a natural generalization of the
classical \textsc{Coloring} problem, which corresponds to the case $\D=0$. The
problem was introduced more than thirty years ago
\cite{andrews1985generalization,CCW86}, and since then has attracted a great
deal of attention
\cite{AchuthanAS11,AngeliniBLD0KMR17,Archdeacon87,BorodinKY13,ChoiE16,CowenGJ97,FrickH94,GoddardX16,HavetS06,KangM10,KimKZ14,KimKZ16}.

From the point of view of applications, \DC\ is particularly interesting in the
context of wireless communication networks, where the assignment of colors to
vertices often represents the assignment of frequencies to communication nodes.
In many practical settings, the requirement of traditional coloring that all
neighboring nodes receive distinct colors is too rigid, as a small amount of
interference is often tolerable, and may lead to solutions that need
drastically fewer frequencies. \DC\ allows one to model this tolerance through
the parameter $\D$. As a result the problem's complexity has been
well-investigated in graph topologies motivated by such applications, such as
unit-disk graphs and various classes of grids
\cite{AraujoBGHMM12,ArchettiBHCG15,Bang-JensenH15,BermondHHS10,GudmundssonMS16,HavetKS09}.
For more background we refer to \cite{frick1993survey,Kang08}.

In this paper we study \DC\ from the point of view of parameterized complexity
\cite{CyganFKLMPPS15,DowneyF13,FlumG06,Niedermeier06}.  The problem is of
course NP-hard, even for small values of $\C,\D$, as it generalizes
\textsc{Coloring}.  We are therefore strongly motivated to bring to bear the
powerful toolbox of structural graph parameters, such as treewidth, which have
proved extremely successful in tackling other intractable hard problems.
Indeed, \textsc{Coloring} is one of the success stories of this domain, since
the complexity of this flagship problem with respect to treewidth (and related
parameters pathwidth, feedback vertex set, vertex cover) is by now extremely
well-understood \cite{LokshtanovMS11a,JaffkeJ17}. We pose the natural question
of whether similar success can be achieved for \DC, or whether the addition of
$\D$ significantly alters the complexity behavior of the problem.  Such results
are not yet known for \DC, except for the fact that it was observed in
\cite{BelmonteLM17} that the problem admits (by standard techniques) a roughly
$(\C\D)^\tw$-time algorithm, where $\tw$ is the graph's treewidth. In
parameterized complexity terms, this shows that the problem is FPT
parameterized by $\tw+\D$. One of our main motivating questions  is whether
this running time can be improved qualitatively (is the problem FPT
parameterized only by $\tw$?) or quantitavely.

Our first result is to establish that the problem is W-hard not just for
treewidth, but also for several much more restricted structural graph
parameters, such as pathwidth, tree-depth, and feedback vertex set. We recall
that for \textsc{Coloring}, the standard $\C^\tw$ algorithm is FPT by $\tw$, as
graphs of bounded treewidth also have bounded chromatic number (Lemma
\ref{lem:basic}). Our result shows that the complexity of the problem changes
drastically with the addition of the new parameter $\D$, and it appears likely
that $\tw$ must appear in the exponent of $\D$ in the running time, even when
$\D$ is large.  More strongly, we establish this hardness even for the case
$\C=2$, which corresponds to the problem of partitioning a graph into two parts
so as to minimize their maximum degree.  This identifies \DC\ as another member
of a family of generalizations of \textsc{Coloring} (such as \textsc{Equitable
Coloring} or \textsc{List Coloring}) which are hard for treewidth
\cite{FellowsFLRSST11}.

As one might expect, the W-hardness results on \DC\ parameterized by treewidth
(or pathwidth, or tree-depth) easily carry over for values of $\C$ larger than
$2$.  Surprisingly, we show that this is \emph{not} the case for the parameter
feedback vertex set, for which the only W-hard case is $2$-coloring: we
establish with a simple win/win argument that the problem is FPT for any other
value of $\C$. We also show that if one considers sufficiently restricted
parameters, such as vertex cover, the problem does eventually become FPT.

Our second step is to enhance the W-hardness result mentioned above with the
aim of determining as precisely as possible the complexity of \DC\
parameterized by treewidth. Our reduction for tree-depth and feedback vertex
set is quadratic in the parameter, and hence implies that no algorithm can
solve the problem in time $n^{o(\sqrt{\tw})}$ under the Exponential Time
Hypothesis (ETH) \cite{IPZ01}.  We therefore present a second reduction, which
applies only to pathwidth and treewidth, but manages to show that no algorithm
can solve the problem in time $n^{o(\pw)}$ or $n^{o(\tw)}$ under the ETH. This
lower bound is tight, as it matches asymptotically the exponent given in the
algorithm of \cite{BelmonteLM17}.

To complement the above results, we also consider the problem from the point of
view of (parameterized) approximation. Here things become significantly better:
we give an algorithm using a technique of \cite{Lampis14} which for any $\C$
and error $\epsilon>0$ runs in time $(tw/\epsilon)^{O(\tw)}n^{O(1)}$ and
approximates the optimal value of $\D$ within a factor of $(1+\epsilon)$.
Hence, despite the problem's W-hardness, we produce a solution arbitrarily
close to optimal in FPT time.

Motivated by this algorithm we also consider the complementary approximation
problem: given $\D$ find a solution that comes as close to the minimum number
of colors needed as possible. By building on the approximation algorithm for
$\D$, we are able to present a $(\tw)^{O(\tw)}n^{O(1)}$ algorithm that achieves
a $2$-approximation for this problem. One can observe that this is not far from
optimal, since an FPT algorithm with approximation ratio better than $3/2$
would contradict the problem's W-hardness for $\C=2$. However, this simple
argument is unsatisfying, because it does not rule out algorithms with a ratio
significantly better than $3/2$, if one also allows a small additive error;
indeed, we observe that when parameterized by feedback vertex set the problem
admits an FPT algorithm that approximates the optimal $\C$ within an additive
error of just $1$.  To resolve this problem we present a gap-introducing
version of our reduction which, for any $i$ produces an instance for which the
optimal value of $\C$ is either $2i$, or at least $3i$.  In this way we show
that, when parameterized by tree-depth, pathwidth, or treewidth, approximating
the optimal value of $\C$ better than $3/2$ is ``truly'' hard, and this is not an
artifact of the problem's hardness for $2$-coloring.

\begin{table}

\centering

\begin{tabular}{|p{2cm}|p{3.5cm}|l|p{3.5cm}|l|} 
\hline
Parameter & Result (Exact solution) & Ref. & Result (Approximation) & Ref. \\
\hline
Feedback\hspace{0.5cm}  Vertex Set & W[1]-hard for $\C=2$ & Thm \ref{thm:fvs1} & $+1$-approximation in time $\fvs^{O(\fvs)}$ & Cor \ref{thm:fvs-approx}  \\
 & FPT for $\C\neq 2$ & Thm \ref{thm:alg-fvs}&& \\
\hline
Tree-depth & W[1]-hard for any $\C\ge2$ & Thm \ref{thm:fvs1} & W[1]-hard to color with $(3/2-\epsilon)\C+O(1)$ colors & Thm \ref{thm:apxhard} \\
\hline
Treewidth, Pathwidth & No $n^{o(\pw)}$ or $n^{o(\tw)}$ algorithm under ETH & Thm \ref{thm:eth} & $(1+\epsilon)$-approximation for $\D$  in $(\tw/\epsilon)^{O(\tw)}$& Thm \ref{thm:tw-approx1}\\
 & & & $2$-approximation for $\C$ in $\tw^{O(\tw)}$& Thm \ref{thm:tw-approx2}\\
\hline
Vertex Cover & $\vc^{O(\vc)}$ algorithm & Thm \ref{thm:alg-vc} && \\
\hline

\end{tabular}

\caption{Summary of results. Hardness results for tree-depth imply the same bounds for treewidth and pathwidth. Conversely, algorithms which apply to treewidth apply also to all other parameters.}

\end{table}

\section{Definitions and Preliminaries}

For a graph $G=(V,E)$ and two integers $\C\ge 1$, $\D\ge 0$, we say that $G$
admits a $(\C,\D)$-coloring if one can partition $V$ into $\C$ sets
such that the graph induced by each set has maximum degree at most $\D$. \DC\
is the problem of deciding, given $G,\C,\D$, whether $G$ admits a
$(\C,\D)$-coloring. For $\D=0$ this corresponds to
\textsc{Coloring}.

We assume the reader is familiar with basic notions in parameterized
complexity, such as the classes FPT and W[1]. For the relevant definitions we
refer to the standard textbooks
\cite{CyganFKLMPPS15,DowneyF13,FlumG06,Niedermeier06}.  We rely on a number of
well-known graph measures: treewidth \cite{BodlaenderK08}, pathwidth,
tree-depth  \cite{NesetrilM06}, feedback vertex set, and vertex cover, denoted
respectively as $\tw(G),\pw(G),\td(G),\fvs(G),\vc(G)$, where we drop $G$ if it
is clear from the context.

\begin{lemma}\label{lem:basic}

For any graph $G$ we have $\tw(G)-1 \le \fvs(G) \le \vc(G)$ and $\tw(G) \le
\pw(G) \le \td(G) -1 \le \vc(G)$. Furthermore, any graph $G$ admits a
$(\tw(G)+1,0)$-coloring, a $(\pw(G)+1,0)$-coloring, a $(\td(G),0)$-coloring,
and a $(\fvs(G)+2,0)$-coloring.

\end{lemma}

The Exponential Time Hypothesis (ETH) states that 3-\textsc{SAT} on instances
with $n$ variables and $m$ clauses cannot be solved in time $2^{o(n+m)}$
\cite{IPZ01}. We define the \MCC\ problem as follows: we are given a graph
$G=(V,E)$, a partition of $V$ into $k$ independent sets $V_1,\ldots,V_k$, such
that for all $i\in\{1,\ldots,k\}$ we have $|V_i|=n$, and we are asked if $G$
contains a $k$-clique. It is well-known that this problem is W[1]-hard
parameterized by $k$, and that it does not admit any $n^{o(k)}$ algorithm,
unless the ETH is false \cite{CyganFKLMPPS15}.

\section{W-hardness for Feedback Vertex Set and Tree-depth}\label{sec:whard}

The main result of this section states that deciding if a graph admits a
$(2,\D)$-coloring, where $\D$ is part of the input, is W[1]-hard parameterized
by either $\fvs$ or $\td$. Because of standard relations between graph
parameters (Lemma \ref{lem:basic}), this implies also the same problem's
W-hardness for parameters $\pw$ and $\tw$. As might be expected, it is not hard
to extend our proof to give hardness for deciding if a $(\C,\D)$-coloring
exists, for any constant $\C$, parameterized by tree-depth (and hence, also
treewidth and pathwidth). What is perhaps more surprising is that this cannot
be done in the case of feedback vertex set. Superficially, the reason we cannot
extend the reduction in this case is that one of the gadgets we use in many
copies in our construction has large $\fvs$ if $\C>2$. However, we give a much
more convincing reason in Theorem \ref{thm:alg-fvs} of Section \ref{sec:exact}
where we show that \DC\ is FPT parameterized by $\fvs$ for $\C\ge 3$, and
therefore, if we could extend our reduction in this case it would prove that
FPT=W[1].

The main theorem of this section is stated below. We then present the reduction
in Sections \ref{sec:fvsred0}, \ref{sec:fvsred1}, and give the Lemmata that
imply Theorem \ref{thm:fvs1} in Section \ref{sec:fvsred2}.

\begin{theorem} \label{thm:fvs1}

Deciding if a graph $G$ admits a $(2,\D)$-coloring, where $\D$ is part of the
input, is W[1]-hard parameterized by $\fvs(G)$. Deciding if a graph $G$ admits
a $(\C,\D)$-coloring, where $\C\ge 2$ is any fixed constant and $\D$ is part of
the input is W[1]-hard parameterized by $\td(G)$.

\end{theorem}

\subsection{Basic Gadgets}\label{sec:fvsred0}

Before we proceed, we present some basic gadgets that will be useful in all the
reductions of this paper (Theorems
\ref{thm:fvs1}, \ref{thm:eth}, \ref{thm:apxhard}).  We first define a building
block $\T(i,j)$ which is a graph that can be properly colored with $i$ colors,
but admits no $(i-1,j)$-coloring (similar constructions appears in \cite{HavetS06}).  
We then use this graph to build two gadgets:
the Equality Gadget and the Palette Gadget (Definitions \ref{def:eq} and
\ref{def:palette}).  Informally, for given $\C,\D$, the equality gadget allows
us to express the constraint that two vertices $v_1,v_2$ of a graph must
receive the same color in any valid $(\C,\D)$-coloring. The palette gadget will
be used to express the constraint that, among three vertices $v_1,v_2,v_3$,
there must exist two with the same color. For both gadgets we first prove
formally that they express these constraints (Lemmata \ref{lem:eq} and
\ref{lem:palette}). We then show that, under certain conditions, these gadgets
can be added to any graph without significantly increasing its  
tree-depth or feedback vertex set (Lemmata \ref{lem:eq2} and
\ref{lem:palette2}).

\begin{definition}

Given two integers $i>0, j\ge 0$, we define the graph $\T(i,j)$ recursively as
follows: $T(1,j) = K_1$ for all $j$; for $i>1$, $T(i,j)$ is the graph obtained
by taking $(j+1)$ disjoint copies of $T(i-1,j)$ and adding to the graph a new
universal vertex.

\end{definition}

\begin{lemma}\label{lem:T}

For all $i>0, j\ge 0$ we have: $\T(i,j)$ admits an $(i,0)$-coloring; $\T(i,j)$
does not admit an $(i-1,j)$-coloring; $\td(\T(i,j)) = \pw(\T(i,j)) +1 =
\tw(\T(i,j))+1 = i$.

\end{lemma}

\begin{definition}\label{def:eq}(Equality Gadget) For $i\ge 2, j\ge0$, we
define the graph $Q(u_1,u_2,i,j)$ as follows: $Q$ contains $ij+1$ disjoint
copies of $\T(i-1,j)$ as well as two vertices $u_1,u_2$ which are connected to
all vertices except each other.  \end{definition}

\begin{lemma}\label{lem:eq} Let $G=(V,E)$ be a graph with $v_1,v_2\in V$ and
let $G'$ be the graph obtained from $G$ by adding to it a copy of
$Q(u_1,u_2,\C,\D)$ and identifying $u_1$ with $v_1$ and $u_2$ with $v_2$. Then,
any $(\C,\D)$-coloring of $G'$ must give the same color to $v_1,v_2$.
Furthermore, if there exists a $(\C,\D)$-coloring of $G$ that gives the same
color to $v_1,v_2$, this coloring can be extended to a $(\C,\D)$-coloring of
$G'$.  \end{lemma}

\begin{lemma}\label{lem:eq2} Let $G=(V,E)$ be a graph, $S\subseteq V$, and $G'$
be a graph obtained from $G$ by repeated applications of the following
operation: we select two vertices $v_1,v_2\in V$ such that $v_1\in S$, add a
new copy of $Q(u_1,u_2,\C,\D)$ and identify $u_i$ with $v_i$, for
$i\in\{1,2\}$.  Then $\td(G') \le \td(G\setminus S) + |S|+\C-1$.
Furthermore, if $\C=2$ we have $\fvs(G') \le \fvs(G\setminus S) + |S|$.
\end{lemma}

\begin{definition}\label{def:palette}(Palette Gadget) For $i\ge 3, j\ge 0$ we
define the graph $P(u_1,u_2,u_3,i,j)$ as follows: $P$ contains ${i\choose
2}j+1$ copies of $\T(i-2,j)$, as well as three vertices $u_1,u_2,u_3$ which are
connected to every vertex of $P$ except each other.  \end{definition}

\begin{lemma}\label{lem:palette} Let $G=(V,E)$ be a graph with $v_1,v_2,v_3\in
V$ and let $G'$ be the graph obtained from $G$ by adding to it a copy of
$P(u_1,u_2,u_3,\C,\D)$ and identifying $u_i$ with $v_i$ for $i\in\{1,2,3\}$.
Then, in any $(\C,\D)$-coloring of $G'$ at least two of the vertices of
$\{v_1,v_2,v_3\}$ must share a color. Furthermore, if there exists a
$(\C,\D)$-coloring of $G$ that gives the same color to two of the vertices of
$\{v_1,v_2,v_3\}$, this coloring can be extended to a $(\C,\D)$-coloring of
$G'$.  \end{lemma}

\begin{lemma}\label{lem:palette2} Let $G=(V,E)$ be a graph, $S\subseteq V$, and
$G'$ be a graph obtained from $G$ by repeated applications of the following
operation: we select three vertices $v_1,v_2,v_3\in V$ such that $v_1,v_2\in
S$, add a new copy of $P(u_1,u_2,u_3,\C,\D)$ and identify $u_i$ with $v_i$, for
$i\in\{1,2,3\}$.  Then $\td(G') \le \td(G\setminus S) + |S|+\C-2$.
\end{lemma}

\subsection{Construction}\label{sec:fvsred1}

We are now ready to present a reduction from \MCC. In this section we describe
a construction which, given an instance of this problem $(G,k)$ as well as an
integer $\C\ge 2$ produces an instance of \DC. Recall that we assume that in
the initial instance $G=(V,E)$ is given to us partitioned into $k$ independent
sets $V_1,\ldots, V_k$, all of which have size $n$.  We will produce a graph
$H(G,k,\C)$ and an integer $\D$ with the property that $H$ admits a
$(\C,\D)$-coloring if and only if $G$ has a $k$-clique. In the next section we
prove the correctness of the construction and give bounds on the values of
$\td(H)$ and $\fvs(H)$ to establish Theorem \ref{thm:fvs1}.

In our new instance we set $\D=|E|-{k\choose 2}$. Let us now describe the graph
$H$. Since we will repeatedly use the gadgets from Definitions \ref{def:eq} and
\ref{def:palette}, we will use the following convention: whenever $v_1,v_2$ are
two vertices we have already introduced to $H$, when we say that we add an
equality gadget $Q(v_1,v_2)$, this means that we add to $H$ a copy of
$Q(u_1,u_2,\C,\D)$ and then identify $u_1,u_2$ with $v_1,v_2$ respectively
(similarly for palette gadgets).  To ease presentation we will gradually build
the graph by describing its different conceptual parts. 

\smallskip

\noindent\textbf{Palette Part}: Informally, the goal of this part is to obtain
two vertices ($p_A,p_B$) which are guaranteed to have different colors. This
part contains the following:

\begin{enumerate}

\item \label{it1} Two vertices called $p_A, p_B$ which we will call the main
palette vertices.

\item \label{it2} For all $i\in \{1,\ldots,\D\}$, $j\in \{A,B\}$ a vertex
$p^i_j$.

\item \label{it3} For all $i\in \{1,\ldots,\D\}$, $j\in \{A,B\}$ we add an
equality gadget $Q(p_j,p^i_j)$.

\item \label{it4} An edge between $p_A, p_B$.

\item \label{it5} For all $i\in \{1,\ldots,\D\}$, $j\in \{A,B\}$ an edge from
$p_j$ to $p^i_j$.

\setcounter{count:constr}{\value{enumi}}

\end{enumerate}

\smallskip

\noindent\textbf{Choice Part}: Informally, the goal of this part is to encode a
choice of a vertex in each $V_i$. To this end we make $2n$ choice vertices for
each color class of the original instance. The selection will be encoded by
counting how many of the first $n$ of these vertices have the same color as
$p_A$. Formally, this part contains the following:

\begin{enumerate}

\setcounter{enumi}{\value{count:constr}}

\item \label{it6} For all $i\in \{1,\ldots, k\}$, $j\in \{1,\ldots,2n\}$ the
vertex $c^i_j$.  We call these the choice vertices.

\item \label{it7} For all $i\in \{1,\ldots, k\}$, $j\in \{A,B\}$ the vertex
$g^i_j$. We call these the guard vertices.

\item \label{it8} For all $i\in \{1,\ldots, k\}$, $j\in \{1,\ldots, 2n\}$ edges
between $c^i_j$ and the vertices $g^i_A$ and $g^i_B$.

\item \label{it9} For all $i\in \{1,\ldots, k\}$, $j\in \{A,B\}$ we add an
equality gadget $Q(p_j,g^i_j)$.

\item \label{it10} If $\C\ge 3$, for all $i\in \{1,\ldots, k\}$, $j\in
\{1,\ldots, 2n\}$ we add a palette gadget $P(p_A,p_B,c^i_j)$.

\setcounter{count:constr}{\value{enumi}}

\end{enumerate}

\smallskip

\noindent\textbf{Transfer Part}: Informally, the goal of this part is to
transfer the choices of the previous part to the rest of the graph. For each
color class of the original instance we make $(k-1)$ ``low'' transfer vertices,
whose deficiency will equal the choice made in the previous part, and $(k-1)$
``high'' transfer vertices, whose deficiency will equal the complement of the
same value. Formally, this part of $H$ contains the following:

\begin{enumerate}

\setcounter{enumi}{\value{count:constr}}

\item  \label{it11} For $i,j\in \{1,\ldots,k\}$, $i\neq j$ the vertex $h_{i,j}$
and the vertex $l_{i,j}$.  We call these the high and low transfer vertices. 

\item  \label{it12} For $i,j\in \{1,\ldots,k\}$, $i\neq j$ and for all
$l\in\{1,\ldots,n\}$ an edge from $l_{i,j}$ to $c^i_l$.

\item  \label{it13} For $i,j\in \{1,\ldots,k\}$, $i\neq j$ and for all
$l\in\{n+1,\ldots,2n\}$ an edge from $h_{i,j}$ to $c^i_l$.

\item  \label{it14} For all $i,j\in \{1,\ldots, k\}$, $i\neq j$ we add an
equality gadget $Q(p_A,l_{i,j})$ and an equality gadget $Q(p_A,h_{i,j})$.

\setcounter{count:constr}{\value{enumi}}

\end{enumerate}

\smallskip

\noindent \textbf{Edge representation}: Informally, this part contains a gadget
representing each edge of $G$. Each gadget will contain a special vertex which
will be able to receive the color of $p_B$ if and only if the corresponding
edge is part of the clique. Formally, we assume that all the vertices of each
$V_i$ are numbered $\{1,\ldots,n\}$. For each edge $e$ of $G$, if $e$ connects
the vertex with index $i_1$ from $V_{j_1}$ with the vertex with index $i_2$
from $V_{j_2}$ (assuming without loss of generality $j_1<j_2$) we add the
following vertices and edges to $H$:

\begin{enumerate}

\setcounter{enumi}{\value{count:constr}}

\item  \label{it15} Four independent sets $L^1_e, H^1_e, L^2_e, H^2_e$ with
respective sizes $n-i_1$, $i_1$, $n-i_2$, $i_2$. 

\item  \label{it16} Edges connecting the vertex $l_{j_1,j_2}$ (respectively,
$h_{j_1,j_2},l_{j_2,j_1},h_{j_2,j_1}$) with all vertices of the set $L^1_e$
(respectively the sets $H^1_e, L^2_e, H^2_e$).

\item  \label{it17} A vertex $c_e$, connected to all vertices in $L^1_e\cup
H^1_e \cup L^2_e \cup H^2_e$.

\item  \label{it18} If $\C\ge 3$, for each $v\in L^1_e\cup H^1_e\cup L^2_e\cup
H^2_e\cup \{c_e\}$ we add a palette gadget $P(p_A,p_B,v)$.

\setcounter{count:constr}{\value{enumi}}

\end{enumerate}

Finally, once we have added a gadget (as described above) for each $e\in E$, we
add the following structure to $H$ in order to ensure that we have a sufficient
number of edges included in our clique:

\begin{enumerate}

\setcounter{enumi}{\value{count:constr}}

\item  \label{it19} A vertex $c_U$ (universal checker) connected to all $c_e$
for $e\in E$.

\item  \label{it20} An equality gadget $Q(p_A,c_U)$.

\setcounter{count:constr}{\value{enumi}}

\end{enumerate} 

\smallskip 

\noindent \textbf{Budget-Setting}: Our construction is now almost done, except
for the fact that some crucial vertices have degree significantly lower than
$\D$ (and hence are always trivially colorable). To fix this, we will
effectively lower their deficiency budget by giving them some extra neighbors.
Formally, we add the following:

\begin{enumerate}

\setcounter{enumi}{\value{count:constr}}

\item  \label{it21} For each guard vertex $g^i_j$, with $j\in \{A,B\}$, we
construct an independent set $G^i_j$ of size $\D-n$ and connect it to $g^i_j$. For
each $v\in G^i_j$ we add an equality gadget $Q(p_j,v)$.

\item  \label{it22} For each transfer vertex $l_{i,j}$ (respectively
$h_{i,j}$), we construct an independent set of size $\D-n$ and connect all its
vertices to $l_{i,j}$ (or respectively to $h_{i,j}$). For each vertex $v$ of
this independent set we add an equality gadget $Q(p_A,v)$.

\item  \label{it24} For each vertex $c_e$ we add an independent set of size
$\D$ and connect all its vertices to $c_e$. For each vertex $v$ of this
independent set we add an equality gadget $Q(p_B,v)$.

\setcounter{count:constr}{\value{enumi}}

\end{enumerate}

This completes the construction of the graph $H$.

\subsection{Correctness}\label{sec:fvsred2}

To establish Theorem \ref{thm:fvs1} we need to establish three properties of
the graph $H(G,k,\C)$ described in the preceding section: that the existence of
a $k$-clique in $G$ implies that $H$ admits a $(\C,\D)$-coloring; that a
$(\C,\D)$-coloring of $H$ implies the existence of a $k$-clique in $G$; and
that the tree-depth and feedback vertex set of $G$ are bounded by some function
of $k$. These are established in the Lemmata below.

\begin{lemma}\label{lem:fvsred1} 

For any $\C\ge 2$, if $G$ contains a $k$-clique, then the graph $H(G,k,\C)$
described in the previous section admits a $(\C,\D)$-coloring.

\end{lemma}

\begin{proof}

Consider a clique of size $k$ in $G$ that includes exactly one vertex from each
$V_i$. We will denote this clique by a function $f:\{1,\ldots,k\} \to
\{1,\ldots,n\}$, that is, we assume that the clique contains the vertex with
index $f(i)$ from $V_i$. We produce a $(\C,\D)$-coloring of $H$ as follows:
vertex $p_A$ receives color $1$, while vertex $p_B$ receives color 2.  All
vertices for which we have added an equality gadget with one endpoint
identified with $p_A$ (respectively $p_B$) take color 1 (respectively 2). We
use Lemma \ref{lem:eq} to properly color the internal vertices of the equality
gadgets. 

We have still left uncolored the choice vertices $c^i_j$ as well as the
internal vertices $L^1_e,H^1_e,L^2_e,H^2_e,c_e$ of the edge gadgets. We proceed
as follows: for all $i\in\{1,\ldots,k\}$ we use color $1$ on the vertices
$c^i_l$ such that $l\in\{1,\ldots,f(i)\}\cup \{n+1,\ldots, 2n-f(i)\}$; we use
color $2$ on all remaining choice vertices.  For every $e\in E$ that is
contained in the clique we color all vertices of the sets $L^1_e, H^1_e, L^2_e,
H^2_e$ with color $1$, and $c_e$ with color $2$.  For all other edges we use
the opposite coloring: we color all vertices of the sets $L^1_e, H^1_e, L^2_e,
H^2_e$ with color $2$, and $c_e$ with color $1$. We use Lemma \ref{lem:palette}
to properly color the internal vertices of palette gadgets, since all palette
gadgets that we add use either color $1$ or color $2$ twice in their endpoints.
This completes the coloring.

To see that the coloring we described is a $(\C,\D)$-coloring, first we note
that by Lemmata \ref{lem:eq},\ref{lem:palette} internal vertices of equality
and palette gadgets are properly colored. Vertices $p_A,p_B$ have exactly $\D$
neighbors with the same color; guard vertices $g^i_j$ have exactly $n$
neighbors with the same color among the choice vertices, hence exactly $\D$
neighbors with the same color overall; choice vertices have at most $k$
neighbors of the same color, and we can assume that $k<|E|-{k\choose 2}$; the
vertex $c_U$ has exactly $\D=|E|-{k\choose 2}$ neighbors with color $1$, since
the clique contains exactly $k\choose 2$ edges; all internal vertices of edge
gadgets have at most one neighbor of the same color.  Finally, for the transfer
vertices $l_{i,j}$ and $h_{i,j}$, we note that $l_{i,j}$ (respectively
$h_{i,j}$) has exactly $f(i)$ (respectively $n-f(i)$) neighbors with color $1$
among the choice vertices.  Furthermore, when $i<j$, $l_{i,j}$ (respectively
$h_{i,j}$) has $|L^1_e|$ (respectively $|H^1_e|$) neighbors with color $1$ in
the edge gadgets, those corresponding to the edge $e$ that belongs in the
clique between $V_i$ and $V_j$. But by construction $|L^1_e|=n-f(i)$ and
$|H^1_e|=f(i)$, and with similar observations for the case $j<i$ we conclude
that all vertices have deficiency at most $\D$.  \end{proof}

\begin{lemma}\label{lem:fvsred2} 

For any $\C\ge 2$, if the graph $H(G,k,\C)$ described in the previous section
admits a $(\C,\D)$-coloring, then $G$ contains a $k$-clique.

\end{lemma}

\begin{lemma}\label{lem:fvsred3} 

For any $\C\ge 2$, the graph $H(G,k,\C)$ described in the previous section has
$\td(H) = O(k^2 + \C)$. Furthermore, if $\C=2$, then $\fvs(H) = O(k^2)$.

\end{lemma}

Theorem \ref{thm:fvs1} now follows directly from the reduction we have
described and Lemmata \ref{lem:fvsred1},\ref{lem:fvsred2},\ref{lem:fvsred3}.

\section{ETH-based Lower Bounds for Treewidth and Pathwidth}\label{sec:eth}

In this section we present a reduction which strengthens the results of Section
\ref{sec:whard} for the parameters treewidth and pathwidth. In particular, the
reduction we present here establishes that, under the ETH, the known algorithm
for \DC\ for these parameters is essentially best possible.

We use a similar presentation order as in the previous section, first giving
the construction and then the Lemmata that imply the result. Where possible, we
re-use the gadgets we have already presented. The main theorem of this section
states the following:

\begin{theorem}\label{thm:eth}

For any fixed $\C\ge 2$, if there exists an algorithm
which, given a graph $G=(V,E)$ and parameters $\C,\D$ decides if $G$ admits a
$(\C,\D)$-coloring in time $n^{o(\pw)}$, then the ETH is false.

\end{theorem}

\subsection{Basic Gadgets}

We use again the equality and palette gadgets of Section \ref{sec:whard}
(Definitions \ref{def:eq},\ref{def:palette}). Before proceeding, let us show
that adding these gadgets to the graph does not increase the pathwidth too
much. For the two types of gadget $Q,P$, we will call the vertices
$u_1,u_2(,u_3)$ the endpoints of the gadget.

\begin{lemma}\label{lem:eq-pw} Let $G=(V,E)$ be a graph and let $G'$ be the
graph obtained from $G$ by repeating the following operation: find a copy of
$Q(u_1,u_2,\C,\D)$, or $P(u_1,u_2,u_3,\C,\D)$; remove all its internal vertices
from the graph; and add all edges between its endpoints which are not already
connected.  Then $\tw(G)\le \max\{\tw(G'),\C\}$ and $\pw(G)\le \pw(G')+\C$.
\end{lemma}

\subsection{Construction}

We now describe a construction which, given an instance $G=(V,E),\ k$, of \MCC\
and a constant $\C$ returns a graph $H(G,k,\C)$ and an integer $\D$ such that
$H$ admits a $(\C,\D)$-coloring if and only if $G$ has a $k$-clique, and the
pathwidth of $H$ is $O(k+\C)$. We use $m$ to denote $|E|$, and we set
$\D=m-{k\choose 2}$. As in Section \ref{sec:whard} we present the construction
in steps to ease presentation, and we use the same conventions regarding adding
$Q$ and $P$ gadgets to the graph.

\smallskip

\noindent \textbf{Palette Part}: This part repeats steps \ref{it1}-\ref{it5} of
the construction of Section \ref{sec:whard}. We recall that this creates two
main palette vertices $p_A,p_B$ (which are eventually guaranteed to have
different colors).

\smallskip

\noindent \textbf{Choice Part}: In this part we construct a sequence of
independent sets, arranged in what can be thought of as a $k\times 2m$ grid.
The idea is that the choice we make in coloring the first independent set of
every row will be propagated throughout the row. We can therefore encode $k$
choices of a number between $1$ and $n$, which will encode the clique.

\begin{enumerate}

\setcounter{enumi}{5}

\item \label{itb6} For each $i\in\{1,\ldots,k\}$, for each
$j\in\{1,\ldots,2m\}$ we construct an independent set $C_{i,j}$ of size $n$.

\item (Backbone vertices) For each $i\in\{1,\ldots,k\}$, for each
$j\in\{1,\ldots,2m-1\}$, for each $l\in\{A.B\}$ we construct a vertex
$b_{i,j}^l$. We connect $b_{i,j}^l$ to all vertices of $C_{i,j}$ and all
vertices of $C_{i,j+1}$.

\item For each backbone vertex $b_{i,j}^l$ added in the previous step, for
$l\in\{A,B\}$, we add an equality gadget $Q(p_l,b_{i,j}^l)$.

\setcounter{count:constr2}{\value{enumi}}

\end{enumerate}

\smallskip

\noindent \textbf{Edge Representation}: In the $k\times 2m$ grid of independent
sets we have constructed we devote two columns to represent each edge of $G$.
In the remainder we assume some numbering of the edges of $E$ with the numbers
$\{1,\ldots,m\}$, as well as a numbering of each $V_i$ with the numbers
$\{1,\ldots,n\}$. Suppose that the $j$-th edge of $E$, where
$j\in\{1,\ldots,m\}$ connects the $j_1$-th vertex of $V_{i_1}$ to the $j_2$-th
vertex of $V_{i_2}$, where $j_1,j_2\in\{1,\ldots,n\}$ and
$i_1,i_2\in\{1,\ldots,k\}$. We perform the following steps for each such edge.

\begin{enumerate}

\setcounter{enumi}{\value{count:constr2}}

\item  We construct four independent sets $H_j^1,L_j^1,H_j^2,L_j^2$ with
respective sizes $n-j_1,j_1,n-j_2,j_2$.

\item \label{itb10} We construct four vertices $h_j^1,l_j^1,h_j^2,l_j^2$.  We
connect $h_j^1$ (respectively $l_j^1,h_j^2,l_j^2$) with all vertices of $H_j^1$
(respectively $L_j^1,H_j^2,L_j^2$).

\item We connect $h_j^1$ to all vertices of $C_{i_1,2j-1}$, $l_j^1$ to all
vertices of $C_{i_1,2j}$, $h_j^2$ to all vertices of $C_{i_2,2j-1}$, $l_j^2$ to
all vertices of $C_{i_2,2j}$.

\item \label{itb12} We add equality gadgets $Q(p_A,h_j^1), Q(p_A,l_j^1),
Q(p_A,h_j^2), Q(p_A,l_j^2)$.

\item \label{itb13} We add a checker vertex $c_j$ and connect it to all
vertices of $H_j^1\cup L_j^1\cup H_j^2\cup L_j^2$.

\setcounter{count:constr2}{\value{enumi}}

\end{enumerate}

\smallskip

\noindent \textbf{Validation and Budget-Setting}: Finally, we add a vertex that
counts how many edges we have included in our clique, as well as appropriate
vertices to diminish the deficiency budget of various parts of our
construction.

\begin{enumerate}

\setcounter{enumi}{\value{count:constr2}}

\item \label{itb14} We add a universal checker vertex $c_U$ and connect it to
all vertices $c_j$ added in step \ref{itb13}. We add an equality gadget
$Q(p_A,c_U)$.

\item \label{itb15} For every vertex $c_j$ added in step \ref{itb13} we
construct an independent set of size $\D$ and connect all its vertices to
$c_j$. For each vertex $v$ in this set we add an equality gadget $Q(p_B,v)$.

\item For each vertex constructed in step \ref{itb10} ($h_j^1, l_j^1, h_j^2,
l_j^2$), we construct an independent set of size $\D-n$ and connect it to the
vertex. For each vertex $v$ of this independent set we add an equality gadget
$Q(p_A,v)$.

\item \label{itb17} For each backbone vertex $b_{i,j}^l$, with $l\in\{A,B\}$,
we construct an independent set of size $\D-n$ and connect it to $b_{i,j}^l$.
For each vertex $v$ of this independent set we add an equality gadget
$Q(p_l,v)$.

\item \label{itb18} If $\C\ge3$, for each vertex $v$ added in steps
\ref{itb6}-\ref{itb17} we add a palette gadget $P(p_A,p_B,v)$.

\end{enumerate}

\subsection{Correctness}

\begin{lemma}\label{lem:eth1}

For any $\C\ge 2$, if $G$ contains a $k$-clique then the graph $H(G,k,\C)$
described in the previous section admits a $(\C,\D)$-coloring.  \end{lemma}

\begin{lemma}\label{lem:eth2}

For any $\C\ge2$, if the graph $H(G,k,\C)$ described in the previous section
admits a $(\C,\D)$-coloring, then $G$ contains a $k$-clique.  \end{lemma}

\begin{lemma}\label{lem:eth3}

For the graph $H(G,k,\C)$ described in the previous section $\pw(H)=O(k+\C)$.

\end{lemma}

The proof of Theorem \ref{thm:eth} now follows directly from Lemmata
\ref{lem:eth1},\ref{lem:eth2},\ref{lem:eth3}.

\section{Exact Algorithms for Treewidth and Other Parameters}\label{sec:exact}

In this section we present several exact algorithms for \DC. Theorem
\ref{thm:tw-exact} gives a treewidth-based algorithm which can be obtained
using standard techniques. Essentially the same algorithm was already sketched
in \cite{BelmonteLM17}, but we give another version here for the sake of
completeness and because it is a building block for the approximation algorithm
of Theorem \ref{thm:tw-approx1}.  Theorem \ref{thm:alg-fvs} uses a win/win
argument to show that the problem is FPT parameterized by $\fvs$ when $\C\neq
2$ and therefore explains why the reduction presented in Section
\ref{sec:whard} only works for $2$ colors.  Theorem \ref{thm:alg-vc} uses a
similar argument to show that the problem is FPT parameterized by $\vc$ (for
any $\C$). 

\begin{theorem}\label{thm:tw-exact}

There is an algorithm which, given a graph $G=(V,E)$, parameters $\C,\D$, and a
tree decomposition of $G$ of width $\tw$, decides if $G$ admits a
$(\C,\D)$-coloring in time $(\C\D)^{O(\tw)}n^{O(1)}$.

\end{theorem}

\begin{theorem}\label{thm:alg-fvs}

\DC\ is FPT parameterized by $\fvs$ for $\C\neq 2$. More precisely, there
exists an algorithm which given a graph $G=(V,E)$, parameters $\C, \D$, with
$\C\neq2$, and a feedback vertex set of $G$ of size $\fvs$, decides if $G$
admits a $(\C,\D)$-coloring in time $\fvs^{O(\fvs)}n^{O(1)}$.

\end{theorem}

\begin{theorem}\label{thm:alg-vc}

\DC\ is FPT parameterized by $\vc$. More precisely, there exists an algorithm
which, given a graph $G=(V,E)$, parameters $\C,\D$, and a vertex cover of $G$
of size $\vc$, decides if $G$ admits a $(\C,\D)$-coloring in time
$\vc^{O(\vc)}n^{O(1)}$.

\end{theorem}

\section{Approximation Algorithms and Lower Bounds}\label{sec:approx}

In this section we present two approximation algorithms which run in FPT time
parameterized by treewidth. The first algorithm (Theorem \ref{thm:tw-approx1})
is an \emph{FPT approximation scheme} which, given a desired number of colors
$\C$, is able to approximate the minimum feasible value of $\D$ for this value
of $\C$ arbitrarily well (that is, within a factor $(1+\epsilon)$). The second
algorithm, which also runs in FPT time parameterized by treewidth, given a
desired value for $\D$, produces a solution that approximates the minimum
number of colors $\C$ within a factor of $2$. 

These results raise the question of whether it is possible to approximate $\C$
as well as we can approximate $\D$, that is, whether there exists an algorithm
which comes within a factor $(1+\epsilon)$ (rather than $2$) of the optimal
number of colors. As a first response, one could observe that such an algorithm
probably cannot exist, because the problem is already hard when $\C=2$, and
therefore an FPT algorithm with multiplicative error less than $3/2$ would
imply that FPT=W[1]. However, this does not satisfactorily settle the problem
as it does not rule out an algorithm that achieves a much better approximation
ratio, if we allow it to also have a small additive error in the number of
colors. Indeed, as we observe in Corollary \ref{thm:fvs-approx}, it is possible
to obtain an algorithm which runs in FPT time parameterized by feedback vertex
set and has an additive error of only $1$, as a consequence of the fact that
the problem is FPT for $\C\ge 3$. This poses the question of whether we can
design an FPT algorithm parameterized by treewidth which, given a
$(\C,\D)$-colorable graph, produces a coloring with $\rho\C+O(1)$ colors, for
$\rho<3/2$.

In the second part of this section we settle this question negatively by
showing, using a recursive construction that builds on Theorem \ref{thm:fvs1},
that such an algorithm cannot exist. More precisely, we present a
gap-introducing version of our reduction: the ratio between the number of
colors needed to color Yes and No instances remains $3/2$, even as the given
$\C$ increases.  This shows that the ``correct'' multiplicative approximation
ratio for this problem really lies somewhere between $3/2$ and $2$, or in other
words, that there are significant barriers impeding the design of a better than
$3/2$ FPT approximation for $\C$, beyond the simple fact that $2$-coloring is
hard.

\subsection{Approximation Algorithms}

Our first approximation algorithm, which is an approximation scheme for the
optimal value of $\D$, relies on a method introduced in \cite{Lampis14} (see
also \cite{AngelBEL16}), and a theorem of \cite{BodlaenderH98}. The high-level
idea is the following: intuitively, the obstacle that stops us from obtaining
an FPT running time with the dynamic programming algorithm of Theorem
\ref{thm:tw-exact} is that the dynamic program is forced to store some
potentially large values for each vertex. More specifically, to characterize a
partial solution we need to remember not just the color of each vertex in a
bag, but also how many neighbors with the same color this vertex has already
seen (which is a value that can go up to $\D$). The main trick now is to
``round'' these values in order to decrease the number of possible states a
vertex can be found in. To do this, we select an appropriate value $\delta$
(polynomial in $\frac{\epsilon}{\log n}$), and try to replace every value that
the dynamic program would calculate with the next higher integer power of
$(1+\delta)$. This has the advantage of limiting the number of possible values
from $\D$ to $\log_{(1+\delta)}\D \approx \frac{\log \D}{\delta}$, and this is
sufficient to obtain the promised running time. The problem is now that the
rounding we applied introduces an approximation error, which is initially a
factor of at most $(1+\delta)$, but may increase each time we apply an
arithmetic operation as part of the algorithm. To show that this error does not
get out of control we show that in any bag of the tree all values stored are
within a factor $(1+\delta)^h$ of the correct ones, where $h$ is the height of
the bag. We then use a theorem of Bodlaender and Hagerup \cite{BodlaenderH98} which states that any
tree decomposition can be balanced in such a way that its height is at most
$O(\log n)$, and as a result we obtain that all values are sufficiently close
to being correct.

The second algorithm we present in this section (Theorem \ref{thm:tw-approx2})
uses the approximation scheme for $\D$ to obtain an FPT $2$-approximation for
$\C$. The idea here is that, given a $(\C,\D)$-colorable graph, we first
produce a $(\C,(1+\epsilon)\D)$-coloring using the algorithm of Theorem
\ref{thm:tw-approx1}, and then apply a procedure which uses $2$ colors for each
color class of this solution but manages to divide by two the number of
neighbors with the same color of every vertex. This is achieved with a simple
polynomial-time local search procedure.

\begin{theorem}\label{thm:bodlaender} \cite{BodlaenderH98} There is a
polynomial-time algorithm which, given a graph $G=(V,E)$ and a tree
decomposition of $G$ of width $\tw$, produces a tree decomposition of $G$ of
width at most $3\tw+2$ and height $O(\log n)$.  \end{theorem}

\begin{theorem}\label{thm:tw-approx1}

There is an algorithm which, given a graph $G=(V,E)$, parameters $\C,\D$, a
tree decomposition of $G$ of width $\tw$, and an error parameter $\epsilon>0$,
either returns a $(\C,(1+\epsilon)\D)$-coloring of $G$, or correctly concludes
that $G$ does not admit a $(\C,\D)$-coloring, in time
$(\tw/\epsilon)^{O(\tw)}n^{O(1)}$.

\end{theorem}

\begin{lemma}\label{lem:max-cut} There exists a polynomial-time algorithm
which, given a graph with maximum degree $\Delta$, produces a two-coloring of
that graph where all vertices have at most $\Delta/2$ neighbors of the same
color.\end{lemma}

\begin{theorem}\label{thm:tw-approx2}

There is an algorithm which, given a graph $G=(V,E)$, parameters $\C,\D$, and a
tree decomposition of $G$ of width $\tw$, either returns a $(2\C,\D)$-coloring
of $G$, or correctly concludes that $G$ does not admit a $(\C,\D)$-coloring, in
time $(\tw)^{O(\tw)}n^{O(1)}$.

\end{theorem}

\subsection{Hardness of Approximation}

The main result of this section is that $\C$ cannot be approximated with a
factor better than $3/2$ in FPT time (for parameters tree-depth, pathwidth, or
treewidth), even if we allow the algorithm to also have a constant additive
error. We remark that an FPT algorithm with additive error $1$ is easy to
obtain for feedback vertex set (Corollary \ref{thm:fvs-approx}).

\begin{theorem}\label{thm:apxhard}

For any fixed $\C>0$, if there exists an algorithm which,
given a graph $G=(V,E)$ and a $\D\ge 0$, correctly distinguishes between the
case that $G$ admits a $(2\C,\D)$-coloring, and the case that $G$ does not
admit a $(3\C-1,\D)$-coloring in FPT time parameterized by $\td(G)$, then
FPT=W[1].

\end{theorem}

\begin{corollary}\label{cor:apxhard}

For any constants $\delta_1,\delta_2>0$, if there exists an algorithm which,
given a graph $G=(V,E)$ that admits a $(\C,\D)$-coloring and parameters
$\C,\D$, is able to produce a $((\frac{3}{2}-\delta_1)\C+\delta_2,\D)$-coloring
of $G$ in FPT time parameterized by $\td(G)$, then FPT=W[1].

\end{corollary}

\begin{corollary}\label{thm:fvs-approx}

There is an algorithm which, given a graph $G=(V,E)$, parameters $\C,\D$, and a
feedback vertex set of $G$ of size $\fvs$, either returns a
$(\C+1,\D)$-coloring of $G$, or correctly concludes that $G$ does not admit a
$(\C,\D)$-coloring, in time $(\fvs)^{O(\fvs)}n^{O(1)}$.

\end{corollary}

\section{Conclusions}

In this paper we classified the  complexity of \DC\ with respect to some of the
most well-studied graph parameters, given essentially tight ETH-based lower
bounds for pathwidth and treewidth, and explored the parameterized
approximability of the problem. Though this gives a good first overview of the
problem's parameterized complexity landscape, there are several questions worth
investigating next. First, is it possible to make the lower bounds of Section
\ref{sec:eth} even tighter, by precisely determining the base of the exponent
in the algorithm's dependence? This would presumably rely on a stronger
complexity assumption such as the SETH, as in \cite{LokshtanovMS11a}.  Second,
can we determine the complexity of the problem with respect to other structural
parameters, such as clique-width \cite{CourcelleMR00}, modular-width
\cite{GajarskyLO13}, or neighborhood diversity \cite{Lampis12}? For some of
these parameters the existence of FPT algorithms is already ruled out by the
fact that \DC\ is NP-hard on cographs \cite{BelmonteLM17}, however the
complexity of the problem is unknown if we also add $\C$ or $\D$ as a
parameter. Finally, it would be very interesting to close the gap between $2$
and $3/2$ on the performance of the best treewidth-parameterized FPT
approximation for $\C$.



\newpage

\bibliography{defective}



\appendix

\section{Omitted Material}

\subsection{Omitted Preliminaries}

We recall here some standard definitions for the reader's convenience.

A tree decomposition of a graph $G=(V,E)$ is a (rooted) tree $T=(X,I)$ such
that each node of $T$ is a subset of $V$. We call the elements of $X$ bags. $T$
must obey the following constraints: $\forall v\in V\  \exists B\in X$ such
that $v\in B$; $\forall (u,v)\in E\ \exists B\in X$ such that $u,v\in B$;
$\forall v\in V$ the bags of $X$ that contain $v$ induce a connected sub-tree.
The width of a tree decomposition is $\max_{B\in X}|B|-1$, and $\tw(G)$ is the
minimum width of a tree decomposition of $G$. Pathwidth is defined similarly,
except the decomposition is required to be a path instead of a tree. 

For a rooted tree $T$ we define its height as the number of vertices in the
longest path from the root to a leaf, and its completion as the graph obtained
by connecting each node to all of its ancestors. For a graph $G$ we define
$\td(G)$ as the minimum height of any tree whose completion contains $G$ as a
subgraph. An equivalent recursive definition is the following: $\td(K_1)=1$; if
$G$ is disconnected then $\td(G)$ is equal to the maximum tree-depth of $G$'s
connected components; otherwise $\td(G)=1+\min_{v\in V}\td(G[V\setminus{v}])$.

A graph's feedback vertex set (respectively vertex cover) is the smallest set
of vertices whose removal leaves the graph acyclic (respectively edge-less). 

\begin{proof}[Proof of Lemma \ref{lem:basic}]

All stated relations are standard but we recall here the proofs for the sake of
completeness. To obtain $\tw(G)-1\le \fvs(G)$, if $S\subseteq V$ is a feedback
vertex set, we can construct a tree decomposition of $G$ by including all
vertices of $S$ in a tree decomposition (of width $1$) of $G[V\setminus S]$.
$\fvs(G)\le vc(G)$ follows because every vertex cover is also a feedback vertex
set. $\tw(G)\le \pw(G)$ because all path decompositions are also valid tree
decompositions. $\pw(G)\le \td(G)-1$ can be seen by recalling that, if $G$ is
connected $\exists v\in V$ such that $\td(G) = 1+\td(G[V\setminus v])$. We can
now take a path decomposition of $G[V\setminus v]$ and add $v$ to every bag. To
see that $\td(G)\le vc(G)+1$ we observe that $G$ is a subgraph of the rooted
tree we construct if we connect all the vertices of a vertex cover in a path,
and attach all the other vertices to the path's last vertex.

For the coloring statements, we recall that a graph with treewidth $\tw$ is
$(\tw+1)$-degenerate, that is, there exists an ordering of its vertices such
that each vertex has at most $\tw+1$ neighbors among the vertices that precede
it \cite{BodlaenderK08}. To see that $\td(G)$ colors suffice to color $G$ if it
is connected, we recall that $\exists v\in V$ such that $\td(G) =
1+\td(G[V\setminus v])$, use a unique color for $v$ and $\td(G)-1$ for the rest
of the graph. $\fvs(G)+2$ colors are always sufficient to properly color a
graph because we can use distinct colors for the feedback vertex set, and
two-color the remaining forest.  \end{proof}

\subsection{Omitted Proofs from Section \ref{sec:whard}}

\begin{proof}[Proof of Lemma \ref{lem:T}]

We begin by the last statement: clearly $\td(\T(1,j)) = \pw(\T(1,j)) +1 =
\tw(\T(1,j)) + 1= 1$, while it can be seen that $\tw(\T(i,j)) + 1 \le
\pw(\T(i,j))+1\le \td(\T(i,j)) \le 1 + \td(\T(i-1,j))$ by removing the
universal vertex. We also observe that $\td(\T(i,j)) \ge \pw(\T(i,j))+1\ge
\tw(\T(i,j))+1 \ge i$ because $\T(i,j)$ contains a clique of size $i$. The
third statement implies the first by Lemma \ref{lem:basic}.  Finally, to see
that $\T(i,j)$ does not admit an $(i-1,j)$-coloring, we do induction on $i$.
Clearly, $\T(1,j)$ requires at least one color. Suppose now that $\T(i,j)$ does
not admit an $(i-1,j)$-coloring but, for the sake of contradiction, $\T(i+1,j)$
admits an $(i,j)$-coloring. By assumption, each of the $j+1$ copies of
$\T(i,j)$ contained in $\T(i+1,j)$ must be using all $i$ available colors.
Hence, each color appears at least $j+1$ times, which implies that there is no
available color for the universal vertex.  \end{proof}

\begin{proof}[Proof of Lemma \ref{lem:eq}] For the first statement, consider a $(\C,\D)$-coloring of $G'$
and examine the copies of $\T(\C-1,\D)$ contained in the equality gadget added
to $G$. For a set $C\subseteq \{1,\ldots,\C\}$ with size $|C|=\C-1$ we say that
$C$ is contained in a copy of $\T(\C-1,\D)$ if all the colors of $C$ appear in
this copy in the coloring of $G'$. There are ${\C \choose \C-1}= \C$ such sets
of colors $C$, and every copy of $\T(\C-1,\D)$ contains at least one by Lemma
\ref{lem:T}. Hence, the set of colors $C$ that is contained in the largest
number of copies is contained in at least $\lceil \frac{\C\D+1}{\C}\rceil =
\D+1$ copies, therefore all its colors appear at least $\D+1$ times. This means
that $v_1,v_2$ cannot take any of the colors in $C$, and therefore must use the
same color.

For the second statement, recall that by Lemma \ref{lem:T}, $\T(\C-1,\D)$ can
be properly colored with $\C-1$ colors, and $\C-1$ colors are available if
$v_1,v_2$ use the same colors.  \end{proof}

\begin{proof}[Proof of Lemma \ref{lem:eq2}] For the first inequality, we begin by observing that $\td(G') \le
\td(G'\setminus S) + |S|$, so it suffices to show that $\td(G'\setminus S) \le
\td(G\setminus S)+\C-1$. Observe now that in $G'\setminus S$, in every copy of
$Q$ one of the vertices $u_1,u_2$ has been removed. 

By definition, there must exist a rooted tree $T_1$ with $\td(G\setminus S)$
levels such that if we complete the tree (that is, connect each node of $T_1$
to all its descendants), $G\setminus S$ is a subgraph of the resulting graph.
Similarly, there exists a rooted tree $T_2$ with $\C-1$ levels such that
$\T(\C-1,\D)$ is a subgraph of its completion.  We now observe that if we take
$T_1$ and attach to each of its nodes a copy of $T_2$ we have a tree with
$\td(G\setminus S)+\C-1$ levels whose completion contains $G'\setminus S$ as a
subgraph.

%
For the final statement, if $\C=2$ the equality gadgets we have added to $G$
contain copies of $\T(1,\Delta)=K_1$. If we remove $S$ from $G'$, and therefore
remove one endpoint of each equality gadget, these vertices become leaves, and
hence do not affect the size of the graph's minimum feedback vertex set.
Deleting them gives us the graph $G\setminus S$, so we conclude that
$\fvs(G'\setminus S) = \fvs(G\setminus S)$ which, together with the fact that
$\fvs(G')\le \fvs(G'\setminus S) + |S|$ completes the proof.  \end{proof}

\begin{proof}[Proof of Lemma \ref{lem:palette}] For the first statement, consider a $(\C,\D)$-coloring of $G'$
and examine the copies of $\T(\C-2,\D)$ contained in the palette gadget added
to $G$.  For a set $C\subseteq \{1,\ldots,\C\}$ with size $|C|=\C-2$ we say
that $C$ is contained in a copy of $\T(\C-2,\D)$ if all the colors of $C$
appear in this copy in the coloring of $G'$. There are ${\C \choose \C-2} =
{\C\choose 2}$ such sets of colors $C$, and every copy of $\T(\C-2,\D)$
contains at least one by Lemma \ref{lem:T}. Hence, the set of colors $C$ that
is contained in the largest number of copies, is contained in at least $\lceil
\frac{{\C\choose 2}\D+1}{{\C\choose 2}}\rceil = \D+1$ copies, therefore all its
colors appear at least $\D+1$ times. This means that $v_1,v_2,v_3$ cannot take
any of the colors in $C$, and therefore have only two colors available for
them. By pigeonhole principle, two of them must share a color.

For the second statement, recall that by Lemma \ref{lem:T}, $\T(\C-2,\D)$ can
be properly colored with $\C-2$ colors, and $\C-2$ colors are available if
$v_1,v_2,v_3$ use at most two colors.  \end{proof}

\begin{proof}[Proof of Lemma \ref{lem:palette2}] The proof follows along the same lines as the proof of Lemma
\ref{lem:eq2}. First, we observe that $\td(G')\le \td(G'\setminus S)+ |S|$ and
then show that $\td(G'\setminus S)\le \td(G\setminus S) + \C-2$ by taking a
tree $T_1$ with $\td(G\setminus S)$ levels whose completion contains
$G\setminus S$ and attaching to each node a tree $T_2$ with $\C-2$ levels whose
completion contains $\T(\C-2,\D)$.  
\end{proof}

\begin{proof}[Proof of Lemma \ref{lem:fvsred2}] Suppose that we are given a
$(\C,\D)$-coloring $c: V(H) \to \{1,\ldots,\C\}$ of $H$. We first establish
that $c(p_A)\neq c(p_B)$. Indeed, because of the equality gadgets added in Step
\ref{it3} we have $c(p^i_j)=c(p_j)$ for all $i\in\{1,\ldots,\D\}, j\in\{A,B\}$.
Because of the edges added in Step \ref{it5} we then know that $p_A,p_B$ each
has at least $\D$ neighbors with the same color. Therefore, because of the edge
connecting them, we conclude that $c(p_A)\neq c(p_B)$. Without loss of
generality we will assume below that $c(p_A)=1$ and $c(p_B)=2$.

Because of the equality gadget of Step \ref{it20} we have $c(c_U)=1$. Because
$c_U$ has degree $|E|$, we conclude that it has at least $k\choose 2$ neighbors
with color $2$. These correspond to a set $E'\subseteq E$ of edges of the
original graph with $|E'|\ge {k\choose 2}$. We will prove that, in fact, $E'$
induces a $k$-clique in $G$.

Let $e\in E'$ be an edge such that $c(c_e)=2$. This implies that all the
vertices of $L^1_e\cup H^1_e \cup L^2_e\cup H^2_e$ must take color $1$, because
by Step \ref{it24} $c_e$ already has $\D$ neighbors with color $2$. In case
$\C\ge 3$ we have also used here the fact that, by Step \ref{it18}, every
internal vertex of the gadget representing $e$ must take color $1$ or $2$.

Suppose that $e\in E'$ connects the vertex with index $i_1$ in $V_{j_1}$ to the
vertex with index $i_2$ in $V_{j_2}$, $j_1<j_2$. We first show that, for an
$e'\in E$ also connecting $V_{j_1}$ to $V_{j_2}$ it must be that $e'\not\in
E'$. Suppose for contradiction that $e'\in E'$, and let $i_1',i_2'$ be the
indices of the endpoints of $e'$. We observe that $l_{j_1,j_2}$ has at least
$|L^1_e|+|L^1_{e'}| = 2n - i_1 - i_1'$ neighbors with color $1$ in the edge
gadgets, while $h_{j_1,j_2}$ has at least $|H^1_e|+|H^1_{e'}| = i_1 + i_1'$
such neighbors. Both $l_{j_1,j_2}$ and $h_{j_1,j_2}$ had $\D-n$ neighbors of
color $1$ added in Step \ref{it22}. Finally, among the $2n$ choice vertices
$c^{j_1}_j$ which are neighbors of either $l_{j_1,j_2}$ or $h_{j_1,j_2}$ there
are at least $n$ which received color $1$, because all the choice vertices have
colors $1$ or $2$ (Step \ref{it10}) and $g^{j_1}_B$, which has color $2$ (Step
\ref{it9}), is connected to all of them and also has $\D-n$ other neighbors of
color $2$ (Step \ref{it21}). Hence, the total number of vertices in
$N(l_{j_1,j_2})\cup N(h_{j_1,j_2})$ with color $1$ is at least $2n + 2(\D-n) +
n >2\D$, hence one of these two vertices has deficiency higher than $\D$,
contradiction. We conclude that $e'\not\in E'$.

To complete the proof, let us show that the $k\choose 2$ edges of $E'$, each of
which connects a different pair of parts of $V$, are incident on the same
endpoints. Take $e\in E'$ as in the previous paragraph, and $e'\in E'$
connecting vertices with indices $i_1',i_3'$ from the parts $V_{j_1},V_{j_3}$,
for $j_3\neq j_2$. It suffices to show that $i_1=i_1'$. Suppose for
contradiction $i_1\neq i_1'$. Consider now the vertices
$l_{j_1,j_2},h_{j_1,j_2},l_{j_1,j_3},h_{j_1,j_3}$, which, by similar reasoning
as before, have $n-i_1$, $i_1$, $n-i_1'$, $i_1'$ color-$1$ neighbors in the
edge gadgets respectively. If there are strictly more than $i_1$ vertices with
color $1$ among the choice vertices $c^{j_1}_j$, $j\in\{1,\ldots,n\}$, then
$l_{j_1,j_2}$ would have deficiency more than $\D$. If there are strictly more
than $n-i_1$ vertices with color $1$ among the choice vertices $c^{j_1}_j$,
$j\in\{n+1,\ldots,2n\}$, then $h_{j_1,j_2}$ would have deficiency more than
$\D$. Since, by the same reasoning as previously, there are at least $n$
vertices with color $i$ among the choice vertices $c^{j_1}_j$, we conclude that
there are exactly $i_1$ vertices with color $1$ among the $c^{j_1}_j$ for
$j\in\{1,\ldots,n\}$, and exaclty $n-i_1$ such vertices in the rest. We can now
conclude that the only way not to violate the deficiency of $l_{j_1,j_3}$ or
$h_{j_1,j_3}$ is for $i_1=i_1'$.  \end{proof}

\begin{proof}[Proof of Lemma \ref{lem:fvsred3}] We first observe that all
equality and palette gadgets added to the graph (Steps \ref{it3}, \ref{it9},
\ref{it10}, \ref{it14}, \ref{it18}, \ref{it20}-\ref{it24}) have at most one
endpoint outside $\{p_A,p_B\}$. Hence, by Lemmata \ref{lem:eq2},
\ref{lem:palette2}, we can conclude that $\td(H) = td(H'\setminus\{p_A,p_B\}) +
\C + 1$ and, for $\C=2$ we have $\fvs(H)\le \fvs(H'\setminus\{p_A,p_B\}) + 2$,
where $H'$ is the graph we obtain from $H$ if we remove all the equality and
palette gadgets. It therefore suffices to show that $\td(H') = O(k^2)$ and, if
$\C=2$, $\fvs(H') = O(k^2)$.

For both parameters we start by removing from the graph all the guard and
transfer vertices, which are $2k + 2k(k-1) = 2k^2$ in total. We now have that
all vertices $p^i_j$, as well as all choice vertices are isolated. Furthermore,
all vertices added to represent edges, as well as the budget-setting vertices,
form a tree with root at $c_U$ and $3$ levels. We conclude that $H'$ has
$\td(H') \le 2k^2 + 4$ and $\fvs(H') \le 2k^2$.  \end{proof}

\subsection{Omitted Proofs from Section \ref{sec:eth}}

\begin{proof}[Proof of Lemma \ref{lem:eq-pw}] First, we observe that there is a path decomposition of
$Q(u_1,u_2,\C,\D)$ with width $\C$, as by Lemma \ref{lem:T} there is a path
decomposition of $\T(\C-1,\D)$ of width $\C-2$, and we can add to all its bags
the vertices $u_1,u_2$. Call this path decomposition $T_Q$. In the same way,
there is a path decomposition of width $\C$ for $P(u_1,u_2,u_3,\C,\D)$, call it
$T_P$. 

We now take an optimal tree or path decomposition of $G'$, call it $T'$, and
construct from it a decomposition of $G$.  Consider a gadget $H\in \{Q,P\}$
that appears in $G$ with endpoints $u_1,u_2(,u_3)$. Since in $G'$ these
endpoints form a clique, there is a bag in $T'$ that contains all of them. Let
$B$ be the smallest such bag. Now, if $T'$ is a tree decomposition, we take
$T_H$ and attach it to $B$.  If $T'$ is a path decomposition, we insert in the
decomposition immediately after $B$ the decomposition $T_H$ where we have added
all vertices of $B$ in all bags of $T_H$. It is not hard to see that in both
cases the decompositions remain valid, and we can repeat this process for every
$H$ until we have a decomposition of $G$. \end{proof}

\begin{proof}[Proof of Lemma \ref{lem:eth1}]
Suppose that $G$ has a $k$-clique, given by a function $\sigma:
\{1,\ldots,k\}\to \{1,\ldots,n\}$, meaning that the clique contains vertex
$\sigma(i)$ from the set $V_i$. We color $H$ as follows: $p_A$ receives color
$1$, $p_B$ receives color $2$, and all vertices on which we have attached
equality gadgets receive the appropriate color, according to Lemma
\ref{lem:eq}. By Lemmata \ref{lem:eq},\ref{lem:palette} we can extend this
coloring to the internal vertices of equality and palette gadgets. For every
independent set $C_{i,j}$, we color $\sigma(i)$ of its vertices with $1$ if $j$
is odd, otherwise we color $n-\sigma(i)$ of its vertices with $1$; we color the
remaining vertices of independent sets $C_{i,j}$ with $2$. For the $j$-th edge
of $E$, if it is contained in the clique then we color $c_j$ with $2$ and
$H_j^1,L_j^1,H_j^2,L_j^2$ with $1$, otherwise we color $c_j$ with $1$ and
$H_j^1,L_j^1,H_j^2,L_j^2$ with $2$. This completes the coloring.

To see that this coloring is valid, observe that the vertices in the palette
part have each at most $\D$ neighbors of the same color; the backbone vertices
$b_{i,j}^l$ have exactly $\D$ neighbors of the same color ($\sigma(i)$ in one
grid independent set and $n-\sigma(i)$ in the other, plus $\D-n$ from step
\ref{itb17}); the vertices $l_j^1,h_j^1,l_j^2,h_j^2$ if the $j$-th edge belongs
to the clique have exactly $\D$ neighbors with the same color; the same
vertices for an edge that does not belong to the clique have strictly fewer
than $\D$ neighbors of the same color; all vertices $c_j$ have at most $\D$
neighbors with the same color; and vertex $c_U$ has $m-{k\choose 2} = \D$
neighbors with the same color.  \end{proof}

\begin{proof}[Proof of Lemma \ref{lem:eth2}] Suppose that we have a valid
$(\C,\D)$-coloring of $H$. As in Lemma \ref{lem:fvsred2} we can assume that
$p_A,p_B$ receive distinct colors, without loss of generality, colors $1$ and
$2$ respectively. Because of step \ref{itb18} we can assume that all the main
vertices of the graph also receive colors $1$ or $2$. Because of the equality
gadget added in \ref{itb14} we know that vertex $c_U$ received color $1$. Since
it has $m$ neighbors, there must exist at least $m-\D = {k\choose 2}$ vertices
$c_j$ which received color $2$.  We call the corresponding edges of $G$ the
selected edges and we will eventually prove that they induce a clique.

We define a set of $k$ vertices of $G$, one from each $V_i$, as follows: in
$V_i$ we select the vertex $\sigma(i)$ if there are $\sigma(i)$ vertices with
color $1$ in $C_{i,1}$. We call these $k$ vertices the selected vertices of
$G$.

We now observe that if there are $\sigma(i)$ vertices with color $1$ in
$C_{i,j}$, then there are $n-\sigma(i)$ vertices with color $1$ in $C_{i,j+1}$.
To see this observe that if there were more than $n-\sigma(i)$ vertices with
color $1$ in $C_{i,j+1}$ this would violate vertex $b_{i,j}^A$, which also has
color $1$ and is connected to $C_{i,j}\cup C_{i,j+1}$. If there were fewer,
this would violate the vertex $b_{i,j}^B$, which has color $2$. Hence, for any
$j\in\{1,\ldots,m\}$ we have that $C_{i,2j-1}$ contains $\sigma(i)$ vertices
with color $1$, while $C_{i,2j}$ contains $n-\sigma(i)$ vertices with color
$1$.

We now want to show that every active edge is incident on two active vertices
to complete the proof. Consider a $c_j$ that corresponds to an active edge.
Since $c_j$ received color $2$, because of step \ref{itb15} all vertices of
$H_j^1,L_j^1,H_j^2,L_j^2$ must have color $1$. Consider now the vertices
$h_j^1,l_j^1$, which also have color $1$ because of step $12$. If $h_j^1$ is
connected to $C_{i_1,2j-1}$ and $l_j^1$ is connected to $C_{i_1,2j}$, then
$h_j^1$ has $(\D-n) + |H_j^1| + \sigma(i_1)$ neighbors with color $1$, while
$l_j^1$ has $(\D-n) + |L_j^1| + n-\sigma(i_1)$ such neighbors. But $|L_j^1| = n
- |H_j^1|$. We therefore have $\sigma(i_1) \le n - |H_j^1|$ as well as
  $\sigma(i_1) \ge |L_j^1| = n-|H_j^1|$.  Therefore, $\sigma(i_1) = |L_j^1|$
and this implies by construction that edge $j$ is incident on vertex
$\sigma(i_1)$ of $V_{i_1}$.  \end{proof}

\begin{proof}[Proof of Lemma \ref{lem:eth3}] We first invoke Lemma
\ref{lem:eq-pw} to replace all palette and equality gadgets with edges. It
suffices to show that the pathwidth of the resulting graph is $O(k)$. We
continue by removing from the graph the vertices $p_A,p_B,c_U$.  This does not
decrease the pathwidth by more than $3$, since these vertices can be added to
all bags. In the remaining graph we remove all leaves and isolated vertices. It
is not hard to see that this does not decrease pathwidth by more than $1$,
since if we find a path decomposition of the remaining graph, we can reinsert
the leaves as follows: for each leaf $v$ we find the smallest bag in the
decomposition that contains its neighbor and insert after it a copy of the same
bag with $v$ added. We note that removing all leaves deletes from the graph all
vertices added for budget-setting, as well as the remaining vertices of the
palette part.

What remains then is to bound the pathwidth of the graph induced by the
backbone vertices $b_{i,j}^l$, the choice vertices in sets $C_{i,j}$, and the
edge representation vertices. We construct a backbone of a path decomposition
as follows: for each $j\in\{1,\ldots,m\}$ we construct a bag that contains all
$b_{i,2j-1}^l, b_{i,2j}^l$, and $b_{i,2j+1}^l$ (if they exist), as well as
$h_j^1, l_j^1, h_j^2, l_j^2, c_j$. We connect these bags in a path in
increasing order of $j$. All these bags have with at most $O(k)$.

We now observe that for every remaining vertex of the graph, there is a bag in
the path decomposition that we have constructed that contains all its
neighbors. We therefore do the following: for every remaining vertex $v$, we
find the smallest bag of the path decomposition that contains its neighborhood,
and insert after it a copy of this bag with $v$ added. This process results in
a valid path decomposition, and it does not increase the size of the largest
bag by more than $1$.  \end{proof}

\subsection{Omitted Proofs from Section \ref{sec:exact}}

\begin{proof}[Proof of Theorem \ref{thm:tw-exact}] The algorithm uses standard
dynamic programming techniques, so we sketch some of the details. We assume we
are given a nice tree decomposition, as defined in \cite{BodlaenderK08}. For
each bag $B_t$ of the decomposition we denote by $B_t^\downarrow$ the set of
vertices included in bags in the sub-tree of the decomposition rooted at $B_t$.
We will maintain in each bag $B_t$ a dynamic programming table $D_t\subseteq
(\{1,\ldots,\C\}\times \{0,\ldots,\D\})^{|B_t|}$. Informally, each element
$s\in (\{1,\ldots,\C\}\times \{0,\ldots,\D\})^{|B_t|}$ is the signature of a
partial solution: we interpret $s$ as a function which, for each vertex in
$B_t$ tells us its color, as well as the number of neighbors this vertex has in
$B_t^\downarrow\setminus B_t$ that share the same color. The invariant we want
to maintain is that $s\in D_t$ if and only if there exists a coloring of
$B_t^\downarrow$ with signature $s$. We can now build the DP table inductively:

\begin{itemize}

\item For a Leaf node $B_t=\{u\}$, $D_t$ contains all signatures $s=(c_u,0)$,
for any $c_u\in\{1,\ldots,\C\}$.

\item For an Introduce node $B_t$ with child $B_{t'}$ such that
$B_t=B_{t'}\cup\{u\}$, for any $s'\in D_{t'}$, and for any
$c_u\in\{1,\ldots,\C\}$, we add to $D_t$ a signature $s$ which agrees with $s'$
on $B_{t'}$ and contains the pair $(c_u,0)$ for vertex $u$.

\item For a Forget node $B_t$ with child $B_{t'}$ such that
$B_t=B_{t'}\setminus\{u\}$ for every signature $s'\in D_{t'}$ we do the
following: let $(c_u,d_u)$ be the pair contained in $s'$ corresponding to
vertex $u$. Let $S_u\subseteq B_{t'}$ be the set of vertices of $B_{t'}$ which
are given color $c_u$ according to $s'$ and which are neighbors of $u$. We
check two conditions: first that $d_u+|S_u|\le \D$; second, that for all $v\in
S_u$ such that $s'$ contains the pair $(c_u,d_v)$ we have $d_v\le \D-1$. If
both conditions hold, we add to $D_t$ a signature $s$ that agrees with $s'$ on
$B_t\setminus S_u$, and that for each $v\in S_u$ such that $s'$ returns
$(c_u,d_v)$, returns the pair $(c_u,d_v+1)$.

\item For a Join node $B_t$ with children $B_{t_1}, B_{t_2}$, (such that
$B_t=B_{t_1}=B_{t_2}$) we do the following: for each $s_1\in D_{t_1}$ and each
$s_2\in D_{t_2}$ we check the following two conditions for all $u\in B_t$: if
$s_1$ returns $(c_{u_1},d_{u_1})$ for $u$ and $s_2$ returns $(c_{u_2},d_{u_2})$
we check if $c_{u_1}=c_{u_2}$; and we check if $d_{u_1}+d_{u_2}\le \D$. If both
conditions hold for all $u\in B_t$ we say that $s_1,s_2$ are compatible, and we
add to $D_t$ a signature $s$ which for $u\in B_t$ contains the pair
$(c_{u_1},d_{u_1}+d_{u_2})$.

\end{itemize}

It is not hard to see that the above operations can be performed in time
polynomial in the size of the table, which is upper-bounded by
$(\C(\D+1))^\tw$. We can then prove by induction that a signature appears in a
table $D_t$ if and only if a coloring with this signature exists for
$B_t^\downarrow$.  If we assume, without loss of generality, that the root bag
contains a single vertex, we can check if the graph admits a $(\C,\D)$-coloring
by checking if the table of the root bag is non-empty.  \end{proof}

\begin{proof}[Proof of Theorem \ref{thm:alg-fvs}] We use a win/win argument.
First, note that we can assume that $\C\ge 3$, since if $\C=1$ the problem is
trivial. Furthermore, if $\C\ge \fvs+2$ then we can produce a
$(\C,\D)$-coloring by giving a distinct color to each vertex of the feedback
vertex set and properly two-coloring the remaining graph. Hence, we assume in
the remainder that $3\le \C\le \fvs+2$.

Now, if $\D\le \fvs$, then we can use the algorithm of Theorem
\ref{thm:tw-exact}. Because of Lemma \ref{lem:basic} this algorithm will run in
time $\fvs^{O(\fvs)}n^{O(1)}$.

Finally, suppose that $\D>\fvs$. In this case the answer is always Yes. To see
this we can produce a coloring as follows: we use a single color for all the
vertices of the feedback vertex set.  Since $\C\ge 3$, there are at least two
other colors available, so we use them to properly color the remaining forest.
This is a valid $(\C,\D)$-coloring, since the only vertices that may have
neighbors of the same color belong in the feedback vertex set, and these can
have at most $\fvs-1<\D$ neighbors with the same color.  \end{proof}

\begin{proof}[Proof of Theorem \ref{thm:alg-vc}] The proof is essentially
identical to that of Theorem \ref{thm:alg-fvs}. We can assume that $\C\le \vc$
(otherwise we use a distinct color for each vertex of the vertex cover, and a
single color for the independent set), and that $\C\ge 2$ (otherwise the
problem is trivial). If $\D\le \vc$ we can use the algorithm of Theorem
\ref{thm:tw-exact}, otherwise we can use a single color for the vertex cover
and another for the independent set.  \end{proof}

\subsection{Omitted Proofs from Section \ref{sec:approx}}

\begin{proof}[Proof of Theorem \ref{thm:tw-approx1}] Our first step is to
invoke Theorem \ref{thm:bodlaender} to obtain a tree decomposition of width
$O(\tw)$ and height $O(\log n)$. We then define a value $\delta =
\frac{\epsilon}{\log^2n}$ and the set $\Sigma = \{0\} \cup \{ (1+\delta)^i\ |
i\in\mathbb{N}, (1+\delta)^i\le (1+\epsilon)\D\}$. In other words, the set
$\Sigma$ contains (in addition to $0$), all positive integer powers of
$(1+\delta)$ with value at most $(1+\epsilon)\D$. We note that $|\Sigma| \le 1+
\log_{(1+\delta)}((1+\epsilon)\D) = O(\log \D/\delta)$, where we have used the
properties $\log_ab = \ln b/\ln a$, and $\ln(1+x)\ge x/2$ for $x$ a
sufficiently small positive constant (that is, for sufficiently large $n$).
Taking into account the value of $\delta$ we have selected, and the fact that
$\D\le n$, we have $|\Sigma| = O(\log^3n/\epsilon)$.

We now follow the outline of the algorithm of Theorem \ref{thm:tw-exact}, with
the difference that we now define a DP table for bag $B_t$ as $D_t \subseteq
(\{1,\ldots,\C\}\times \Sigma)^{|B_t|}$. Again, we interpret the elements of
$D_t$ as functions which, for each vertex in $B_t$ return a color and an
\emph{approximate} number of neighbors that have the same color as this vertex
in $B^\downarrow_t\setminus B_t$.

More precisely, if a bag $B_t$ is at height $h$ (that is, its maximum distance
from a leaf bag in the sub-tree rooted at $B_t$ is $h$) we will maintains the
following two invariants: 

\begin{enumerate}

\item If there exists a coloring $\mathbf{c}$ of $B^\downarrow_t$ such that all
vertices of $B^\downarrow_t\setminus B_t$ have at most $\D$ neighbors of the
same color, and all vertices of $B_t$ have at most $\D$ neighbors of the same
color in $B^\downarrow_t\setminus B_t$, then there exists $s\in D_t$ which
assigns the same colors as $\mathbf{c}$ to $B_t$; and which, if $u\in B_t$ has
$d_u'$ neighbors with the same color in $B^\downarrow_t\setminus B_t$ in
$\mathbf{c}$, returns value $d_u\le (1+\delta)^h d_u'$ for vertex $u$, where
$d'_u\in\Sigma$. 

\item If there exists a signature $s\in D_t$, then there exists a coloring
$\mathbf{c}$ of $B^\downarrow_t$ such that all vertices of
$B^\downarrow_t\setminus B_t$ have at most $(1+\epsilon)\D$ neighbors; all
vertices of $B_t$ take in $\mathbf{c}$ the colors described in $s$; if $s$
dictates that a vertex $u\in B_t$ has $d_u$ neighbors with the same color in
$B^\downarrow_t\setminus B_t$, then $u$ has at most $d_u$ neighbors with the
same color in $B^\downarrow_t\setminus B_t$ according to coloring $\mathbf{c}$.

\end{enumerate}

The first of the two properties above implies that, if there exists a
$(\C,\D)$-coloring of $G$, the algorithm will be able to find some entry in the
table of the root bag that will allows us to construct a
$(\C,(1+\delta)^H)$-coloring, where $H$ is the height of the tree
decomposition. We recall now that $H=O(\log n)$, therefore, $(1+\delta)^H \le
e^{\delta H} \le e^{O(\epsilon/\log n)} \le 1+\epsilon$. Hence, if we establish
the first property, we know that if a $(\C,\D)$-coloring exists, the algorithm
will be able to find a $(\C,(1+\epsilon)\D)$-coloring. Conversely, the second
property assures us that, if the algorithm places a signature $s$ in a DP
table, there must exist a coloring that matches this signature.

In order to establish these invariants we must make a further modification to
the algorithm of Theorem \ref{thm:tw-exact}. We recall that the algorithm makes
some arithmetic calculation in Forget nodes (where the value $d_v$ of neighbors
of the forgotten node with the same color is increased by $1$); and in Join
nodes (where values $d_{u_1},d_{u_2}$ corresponding to the same node are
added). The problem here is that even if the values stored are integer powers
of $(1+\delta)$, the results of these additions are not necessarily such
integer powers. Hence, our algorithm will simply ``round up'' the result of
these additions to the closest integer power of $(1+\delta)$. Formally, instead
of the value $d_v+1$ we use the value $(1+\delta)^{\lceil \log_{(1+\delta)}
(d_v+1) \rceil}$, and instead of the value $d_{u_1}+d_{u_2}$ we use the value
$(1+\delta)^{\lceil \log_{(1+\delta)} (d_{u_1}+d_{u_2}) \rceil}$.

We can now establish the two properties by induction. The two interesting cases
are Forget and Join nodes. For a Join node of height $h$ and the first
property, if we have established by induction that for the two values
$d_{u_1},d_{u_2}$ stored in the children's tables we have $d_{u_1}\le
(1+\delta)^{h-1}d_{u_1}'$, $d_{u_2}\le(1+\delta)^{h-1}d_{u_2}'$, where
$d_{u_1}',d_{u_2}'$ are as described in the first property, then
$d_{u_1}+d_{u_2}\le (1+\delta)^{h-1}(d_{u_1}'+d_{u_2}')$. However, for the new
value we calculate we have $d_u\le (1+\delta) (d_{u_1}+d_{u_2}) \le
(1+\delta)^h (d_{u_1}'+d_{u_2}') = (1+\delta)^h d_u'$. For the second property,
observe that since we always round up, the value stored in the table will
always be at least as high as the true number of neighbors of a vertex in the
coloring $\mathbf{c}$. Calculations are similar for Forget nodes.

Because of the above we have an algorithm that runs in time polynomial in
$|D_t| = (\C |\Sigma|)^{O(\tw)}$. We can assume without loss of generality that
$\C\le \tw+1$, otherwise by Lemma \ref{lem:basic} the graph can be easily
properly colors. By the observations of $|\Sigma|$ we therefore have that the
running time is $(\tw\log n/\epsilon)^{O(\tw)}$. A well-known win/win argument
allows us to obtain the promised bound as follows: if $\tw \le \sqrt{\log n}$,
this running time is in fact polynomial in $n,1/\epsilon$, so we are done; if
$\sqrt{log n} \le tw$ then $\log n \le \tw^2$ and the running time is upper
bounded by $(\tw/\epsilon)^{O(\tw)}$.  \end{proof}

\begin{proof}[Proof of Lemma \ref{lem:max-cut}] We run what is essentially a
local search algorithm for \textsc{Max Cut}.  Initially, color all vertices
with color $1$. Then, as long as there exists a vertex $u$ such that the
majority of its neighbors have the same color as $u$, we change the color of
$u$. We continue with this process until all vertices have a majority of their
neighbors with a different color. In that case the claim follows. To see that
this procedure terminates in polynomial time, observe that in each step we
increase the number of edges that connect vertices of different colors.
\end{proof}

\begin{proof}[Proof of Theorem \ref{thm:tw-approx2}] We assume without loss of
generality that $\D$ is sufficiently large (e.g.  $\D\ge 20$), otherwise we can
solve the problem exactly by using the fact that $\C$ is bounded by $\tw$ (by
Lemma \ref{lem:basic}) and the algorithm of Theorem \ref{thm:tw-exact}. We
invoke the algorithm of Theorem \ref{thm:tw-approx1}, setting $\epsilon=1/10$.
The algorithm runs in the promised running time. If it reports that $G$ does
not admit a $(\C,\D)$-coloring, we output the same answer and we are done.

Suppose that the algorithm of Theorem \ref{thm:tw-approx1} returned a
$(\C,\frac{11}{10}\D)$-coloring of $G$. We transform this to a
$(2\C,\Delta)$-coloring by using Lemma \ref{lem:max-cut}.

We consider each color class in the returned coloring of $G$ separately. Each
class induces a graph with maximum degree $\frac{11}{10}\D$. According to Lemma
\ref{lem:max-cut}, we can two-color this graph so that no vertex has more than
$\frac{11}{20}\D\le \D$ neighbors with the same color.  We produce such a
two-coloring for the graph induced by each color class using two new colors.
Hence, the end result is a $(2\C,\frac{11}{20}\D)$-coloring of $G$, which is
also a valid $(2\C,\D)$-coloring.  \end{proof}

\begin{proof}[Proof of Theorem \ref{thm:apxhard}] First, observe that the
theorem already follows for $\C=1$ by Theorem \ref{thm:fvs1}, which states that
it is W[1]-hard parameterized by $\td(G)$ to decide if a graph admits a
$(2,\D)$-coloring. Let $G^1$ be the graph produced in the reduction of Theorem
\ref{thm:fvs1}. By repeated composition we will construct, for any $\C$, a
graph $G^{\C}$ such that either $G^\C$ admits a $(2\C,\D)$-coloring, or it does
not admit a $(3\C-1,\D)$-coloring, depending on whether $G^1$ admits a
$(2,\D)$-coloring.

Suppose that we have constructed the graph $G^\C$, for some $\C$. We describe
how to build the graph $G^{\C+1}$. We start with a copy of $G^1$, which we call
the main part of our construction. We will add to this many disjoint copies of
$G^\C$ and appropriately connect them to $G^1$ to obtain $G^{\C+1}$.

Recall that the graph $G^1$ contains two palette vertices $p_A,p_B$, each
connected to $\D$ neighbors $p^i_j$, $i\in\{1,\ldots,\D\}$, $j\in\{A,B\}$ with
both edges and equality gadgets. Furthermore, recall that for two colors, an
equality gadget with endpoints $p_j, p^i_j$ is an independent set on $2\D+1$
vertices which are common neighbors of $p_j$ and $p^i_j$.

For each $j\in\{A,B\}$, each $i\in\{1,\ldots,\D\}$, and each internal vertex
$v$ of the equality gadget $Q(p_j,p^i_j)$ added in step \ref{it3} we add to the
main graph ${3\C+2 \choose 3\C}\D+1$ disjoint copies of $G^\C$ and connect all
their vertices to $p_j,p^i_j$, and $v$.

Now, for every vertex $v$ of $G^1$ that is not part of the palette (that is,
every vertex that was not constructed in steps \ref{it1}-\ref{it5}), we add
another ${3\C+2 \choose 3\C}\D+1$ disjoint copies of $G^\C$ and connect all
their vertices to $p_A,p_B$, and $v$.

This completes the construction. We now need to establish three properties:
that if $G^1$ admits a $(2,\D)$-coloring then $G^{\C+1}$ admits a
$(2\C+2,\D)$-coloring; that if $G^1$ does not admit a $(2,\D)$-coloring then
$G^{\C+1}$ does not admit a $(3\C+2,\D)$-coloring; and that the tree-depth of
$G^{\C+1}$ did not increase too much. 

We proceed by induction and assume that all the above have been shown for
$G^\C$. For the first property, if $G^1$ admits a $(2,\Delta)$-coloring and
$G^\C$ admits a $(2\C,\D)$-coloring, then we can construct a coloring of
$G^{\C+1}$ by taking the same coloring with $2\C$ colors for all the copies of
$G^\C$, and using two new colors to color the main graph $G^1$.

For the second property, suppose that we know that a $(3\C-1,\D)$-coloring of
$G^\C$ implies the existence of a $(2,\D)$-coloring of $G^1$. We want to show
that a $(3\C+2,\D)$-coloring of $G^{\C+1}$ also implies a $(2,\D)$-coloring of
$G^1$. Suppose then that we have such a $(3\C+2,\D)$-coloring of $G^{\C+1}$. If
a copy of $G^\C$ included in $G^{\C+1}$ uses at most $3\C-1$ colors, we are
done, since this implies the existence of a $(2,\D)$-coloring of $G^1$.
Therefore, assume that all copies of $G^{\C+1}$ use at least $3\C$ colors.

Consider now two vertices $p_j, p^i_j$, for some $j\in\{A,B\}$,
$i\in\{1,\ldots,\D\}$. We claim that they must receive the same color. To see
this, take an internal vertex $v$ of the equality gadget $Q(p_j,p^i_j)$ and
recall that we have added ${3\C+2\choose 3\C}\D+1$ disjoint copies of $G^\C$
connected to $p_j,p^i_j,v$. Hence, there is some set of $3\C$ colors that
appears in at least $\D+1$ of these copies, and therefore cannot be used in
$p_j,p^i_j,v$. Therefore, if $p_j,p^i_j$ do not share a color, all the $2\D+1$
internal vertices of the equality gadget share the color of one of the two,
which violates the correctness of the coloring. We conclude that $p_A$ has $\D$
neighbors with its own color, as does $p_B$, therefore, since they are
connected, $p_A,p_B$ use distinct colors.

Consider now any other vertex $v$ of the main graph. Again, we have added
${3\C+2\choose 3\C}\D+1$ disjoint copies of $G^\C$ connected to $p_A,p_B,v$,
hence there is a set of $3\C$ colors which appears in $\D+1$ copies and is
therefore not used by $p_A,p_B,v$. Since there are $3\C+2$ colors overall and
$p_A,p_B$ use distinct colors, we conclude that $v$ uses either the color of
$p_A$ or that of $p_B$. Hence, the coloring of $G^{\C+1}$ contains a 2-coloring
of $G^1$.

For the final property, suppose that $\td(G^\C) \le \C \td(G^1) + 2\C$. We want
to establish that $\td(G^{\C+1}) \le (\C+1)\td(G^1) +2\C + 2$. To see this, we
construct a tree for $G^{\C+1}$ as follows, the two top vertices are $p_A,p_B$,
and below these we place a tree whose completion contains $G^1$ (hence we have
at most $\td(G^1)+2$ levels now). For every copy of $G^\C$ that was connected
to $p_A,p_B$, and a vertex $v$, we find $v$ and attach below it a tree whose
completion contains $G^\C$. Similarly, for every copy of $G^\C$ attached to
$p_j,p^i_j$, and a vertex $v$, for some $j\in\{A,B\}$, $i\in\{1,\ldots,\D\}$,
one of the vertices $v,p^i_j$ is a descendant of the other in the current tree
(since they are connected); we attach a tree containing $G^\C$ to this
descendant. The total number of levels of the tree is therefore $\td(G^1)+2 +
\td(G^\C) \le (\C+1)\td(G^1)+\C+2$, as desired.  \end{proof}

\begin{proof}[Proof of Corollary \ref{cor:apxhard}] Fix some constants
$\delta_1,\delta_2$. We invoke Theorem \ref{thm:apxhard} with $\C = \lceil
\frac{\delta_2+1}{\delta_1}\rceil$. The graph produced either admits a
$(2\C,\Delta)$-coloring or does not admit a $(3\C-1,\Delta)$-coloring.  Suppose
that the algorithm described in this corollary exists. Then, in the former case
it produces a coloring with at most $(\frac{3}{2}-\delta_1)\cdot 2\lceil
\frac{\delta_2+1}{\delta_1}\rceil + \delta_2 = 3\lceil
\frac{\delta_2+1}{\delta_1}\rceil - 2\delta_1 \lceil
\frac{\delta_2+1}{\delta_1}\rceil + \delta_2 \le 3\C - 2(\delta_2+1) + \delta_2
\le 3\C-1$ colors. Hence, the algorithm would be able to distinguish the two
cases of a W[1]-hard problem.  \end{proof}

\begin{proof}[Proof of Corollary \ref{thm:fvs-approx}] If $\C\ge 3$ we simply
invoke Theorem \ref{thm:alg-fvs}. If $\C=2$ we invoke the same algorithm with
$\C=3$. If the algorithm produces a coloring, we output that as the solution,
otherwise we can report that no $(\C,\D)$-coloring exists.  \end{proof}

\end{document}